\documentclass{article}

\usepackage[main,final,nonatbib]{neurips_2026}

\usepackage[utf8]{inputenc} 
\usepackage[T1]{fontenc}    
\usepackage{hyperref}       
\usepackage{url}            
\usepackage{booktabs}       
\usepackage{amsfonts}       
\usepackage{nicefrac}       
\usepackage{microtype}      
\usepackage{xcolor}         

\usepackage{amsmath}
\usepackage{amsthm}
\usepackage{caption}
\usepackage{subcaption}
\usepackage{graphicx}
\usepackage{cleveref}
\usepackage[backend=biber,style=numeric]{biblatex}
\newtheorem{proposition}{Proposition}

\title{Voronoi-Markov chain and spatial entropy for point pattern analysis}

\author{%
  James~C.~Mathews~Jr.\\
  Department of Medical Physics\\
  Memorial Sloan Kettering Cancer Center
  \\
  New York NY USA\\
  \texttt{mathewj2@mskcc.org} \\
  \And
  Francisco~Couzo \\
  DigITs AI/ML Solutions\\
  Memorial Sloan Kettering Cancer Center\\
  New York, NY, USA\\
  \And
  Aleksandr~Petrov\\
  DigITs AI/ML Solutions\\
  Memorial Sloan Kettering Cancer Center\\
  New York, NY, USA\\
  \And
  Avijit~Chatterjee \\
  DigITs AI/ML Solutions\\
  Memorial Sloan Kettering Cancer Center\\
  New York, NY, USA\\
  \And
  Saad~Nadeem\\
  Department of Medical Physics\\
  Memorial Sloan Kettering Cancer Center\\
  New York, NY USA\\
  \texttt{nadeems@mskcc.org} \\
}

\begin{document}

\maketitle

\begin{abstract}
In this paper, we revisit the concept of entropy in the spatial context with the aim of deriving computable and interpretable metrics for point pattern analysis in domains such as histopathology. We discuss stochastic point process assumptions like Poisson homogeneity and Simple Sequential Inhibition (SSI) and review established results for Voronoi and Delaunay tilings of point sets for their potential to provide null distributions for rigorous empirical hypothesis testing. We present (1) a novel ``dense basin entropy'' defined in terms of the Ord distribution supported on Voronoi basins, which is shown to be sensitive to clustering, and (2) a related Markov chain designed as a simplification or approximation of Brownian motion in the underlying planar domain. We investigate bounds on the dense basin entropy in the SSI regime and in empirical data, and show that the entropy rate and spectral gap of the Markov chain lead to sharp discrimination metrics. Finally, we demonstrate a generalization to multiple-set analysis via aggregation of one set over the Markov invariant measure for another. Simulations and experiments in histopathology show that our proposed metrics offer insights complementary to other instantiations of entropy in spatial analysis. The code can be found on our SMProfiler GitHub: \url{https://github.com/nadeemlab/SMProfiler}.
\end{abstract}

\section{Introduction}
Quantifying and stratifying patterns in spatial point sets is a key component of data analysis in the empirical sciences, for example ecological field studies, geography, epidemiology, and a domain of particular interest in computer vision, histopathology (with segmented cells represented as point sets). Discriminating random from non-random patterns is especially important for discovery, but also difficult due to the open-ended nature of the problem.

We aim to apply the concept of entropy, from the theory of codes and information transmission, to such point sets. It is reasonable to expect a form of entropy to help discriminate between random point sets and those exhibiting systematic patterns. This approach has the advantage that it does not require much prior knowledge of the kind of pattern, making it suitable for exploratory data analysis in data modalities thought to carry a lot of information that is not yet well-structured. To this end, reviewing previous work on spatial generalizations of entropy, we find that it calls for principled/adapted discretizations and that the Voronoi/Thiesson tiling of a given point set, and its dual, Delaunay triangulation, can be used to frame this construction.

In cell microscopy, there are several established metrics used for point pattern analysis, including spatial autocorrelation (Moran's I and Geary's C \cite{pebesma2023spatial}), the Ripley statistics \cite{baddeley2016spatial}, and neighborhood enrichment scores \cite{palla2022squidpy}. Often an appropriate null or reference assumption is not explicit, or the distribution of the computed metrics under the null assumption is not established, and for this reason such metrics are usually limited to use within a supervised learning paradigm. In traditional statistics much effort is expended in the determination of rigorous null distributions, in order to make unsupervised learning successful.

We found that the literature on the Delaunay triangulation attributes under the Poisson point process assumption (homogeneity), going back to the 1970s, is quite comprehensive and offers a model example of a completely described null reference for empirical hypothesis testing. This enables unsupervised analysis for point patterns while making the case that the Poisson-Delaunay statistics should be used more widely for this reason.

As the spatial entropy and Delaunay statistics are defined in terms of aggregation of isolated local geometric structures, they do not assess interaction between elements or influence across multiple scales. To rectify this, we introduce the Voronoi-Markov chain as a generalization of our spatial entropy. We evaluate the various metrics in simulations intended to highlight the sensitivity to variations in point set density.

\section{Related Work}
The concept of entropy was applied extensively to spatial data analysis following the formalization and justification of entropy in coding theory \cite{shannon1948mathematical}. The applications to settlement pattern analysis \cite{medvedkov1967concept, medvedkov1967regular} are early examples from 1967. Initial applications were focused on regular grids, a limitation highlighted by several follow up works \cite{batty1972entropy}. Problems with entropy formulations based on space discretizations were frequently noted (e.g. \cite{jaynes1957information,batty1974spatial}), and differential entropy defined for continuous distributions in space has not enjoyed widespread acceptance and use, despite reasonable justification in \cite{shannon1948mathematical} and the eventual availability of a rigorous estimator of differential entropy from samples \cite{kozachenko1987sample}.

In seeking to generalize ordinary discrete entropy to the spatial context, one approach is to retain the discrete categorical variable and augment samples with a spatial point set component, for the context of categorized point sets. This avoids the issue of defining differential entropy and leads to modifications of the normal formula in which terms or factors are added to account for inter-categorical influence evinced by spatial arrangement. A recent review article \cite{li2025entropy} shows that the variety of such entropies is rich with untapped potential for applications, although a theory that clearly organizes them mathematically is not yet available. The main limitation of this approach is that it becomes trivial when the number of categories is small, so it is not informative for ordinary point pattern analysis (for one set) or multiple-set analysis with 2 or 3 sets.

By contrast, the original approach in the earlier applications was to \emph{replace} the categorical variable with a location variable. The issue of differential entropy was again circumvented, by working directly with a discretized location. The first use of the Voronoi/Thiessen tilings for precisely this purpose was reported in 1970 by Chapman in an application to geography \cite{chapman1970application} (see also \cite[p99]{upton1985spatial}), using what we call the ``uniform basin entropy''.

The Voronoi polygons of a point set $\{v_i\}$ are the regions consisting of points closest to one of the $v_i$ rather than to any other. Due to its naturalness and non-triviality, this construction has a rich history of rediscovery, and the reader is referred to \cite{okabe1999spatial} for details regarding efficient algorithms, generalizations, and analytical results, with comprehensive references. The dual, Delaunay tiling, being a triangulation, is more amenable to precise analysis. The independent findings of Collins \cite{collins1968geometrical} and Miles \cite{miles1970homogeneous} initiated a tradition of establishing rigorous analytical properties of random Delaunay triangles and eventually also Voronoi polygons (see \cite{okabe1999spatial} and \cite{moller1998review}). These properties include precise formulas for the distribution of the angles of the Delaunay triangles, and several authors make the case for the use of these and other distributions for empirical hypothesis testing of specific point sets.

Following the works \cite{vincent1982theoretical,boots1986using}, few studies have carried out this prescription in specific spatial data analysis applications, \textit{partly because the related body of literature narrowly predates the onset of the era of cheap computing in the 1990s and 2000s}. By the time computational resources for analysis became readily available in parallel with large-scale data collections e.g. in the form of high-resolution rasterized imaging (histopathology/satellite imagery), the field had already started transitioning towards more brute-force data mining techniques and the antecedents of machine learning.

Against this backdrop, there is now renewed interest in theoretically justified techniques in the form of metrics that can be practically computed and meaningfully interpreted, e.g. the Squidpy Python library for computing spatial statistics in spatial -omics \cite{palla2022squidpy}. The Delaunay statistics seem to have been overlooked in recent research for this purpose, despite being well suited to it. A notable exception to the trend is the 2024 theoretical study of Edelsbrunner et al. \cite{edelsbrunner2024angles} which generalizes the classical Delaunay statistics to Brillouin tessellations of order higher than 1.

In another direction, the idea to view a Voronoi tiling as a frame for a Markov process to simplify a continuous model also appears in \cite{vanden2009markovian,jagger2020predicting} in the context of molecular dynamics and rigorous estimation of chemical kinetics rates. These models are sophisticated and they make use of a Voronoi tiling not for an observed point set as in pattern analysis, but for a subset tuned in order to achieve optimal characteristics for stochastic modeling.

\section{Methods}

\subsection{Entropy Formalism}

\paragraph{Entropy of a Discrete Distribution.}
Ordinary entropy is defined for a finite set of probabilities $p_i$, $i=1, ..., N$, with $\sum p_i = 1$, as $H:=\sum p_i (- \operatorname{log} p_i )$. We default to the use of base $e$ logarithms throughout, without much loss of generality, but base $2$ is also commonly used. Setting $F(x):=x(-\operatorname{log} x)$, we find that $H=\sum F(p_i)$ is a sum of terms belonging to $[0,e^{-1}]$, as $F$ is a concave positive function on $[0,1]$, differentiable to all orders on $(0,1]$, vanishing at $0$ and $1$, and with maximum value at $(e^{-1}, e^{-1})$. For concreteness, some basic properties are listed below.

\begin{proposition} Fix the positive integer $N$ and consider $H$ as a function on the simplex $\Delta^{N-1}=\{(p_i) \vert \thinspace \sum p_i = 1, \thinspace 0\leq p_i\leq 1 \}$.
  \begin{enumerate}
  \setlength\itemsep{-0.3em}
    \item{ \label{bounds} $H\in [0, \operatorname{log} N]$. }
    \item{ \label{attains} $H$ attains the maximum value, $\operatorname{log} N$, for the uniform distribution $p_i=1/N$.}
    \item{ \label{stationarypoint} The uniform distribution is the only stationary point for $H$ on the interior of $\Delta^{N-1}$.}
    \item{ \label{midpoints} $H$ takes the value $\operatorname{log} (d+1)$ at the midpoint/centroid of each $d$-dimensional facet $\Delta^{d}$ of $\Delta^{N-1}$.}
  \end{enumerate}
\end{proposition}

\noindent
\emph{Proof}. (\ref{bounds}), (\ref{attains}): See \cite[Theorem 1.4 p11]{polyanskiy2025information}, \cite[p3]{csiszar2011information}, \cite[p11]{shannon1948mathematical}. (\ref{stationarypoint}): The method of Lagrange multipliers for constraint $1-\sum p_i=0$ shows that at a stationary point for $H$, all $F'(p_i)$ are equal. $F'(x)=-(1+\operatorname{log} x)$ is monotonic bijective on $(0,1)$, so all $p_i$ must be equal. (\ref{midpoints}): This follows from the consistency of the expression of $H$ for given $N$ with the expression for smaller $N$ values as some $p_i$ decrease to zero. \qedsymbol{}

\paragraph{Justification of Discrete Entropy.}

The significance of $H$ is due in large part to its interpretation, in the context of long sequences of characters drawn from a set of $N$ symbols occurring with frequencies consistent with $p_i$, as the average number of information units (i.e. bits in the base 2 case) needed per character to represent the sequences with a variable length code (\cite[Theorem 1.1 p3]{csiszar2011information}, \cite[Theorem 3 p13]{shannon1948mathematical}). Thus a message stream exhibiting the distribution $p_i$ may be said to produce information at the rate $H$ bits per character.

Ellerman recently introduced a formulation of information theory emphasizing the distinctions made by a partition rather than the probability of membership in a set \cite{ellerman2021new}. One of this theory's contributions is the \emph{logical entropy}, which is defined similarly to the formula for $H$ above but with the function $F$ replaced by $F_l(x)=x(1-x)$. Ellerman shows that the logical entropy has a number of favorable properties and in what follows we will attempt to consider both Shannon and logical entropy whenever entropy is invoked.

\paragraph{The Issue of Continuous/Differential Entropy and Discretization.}
For any application to a spatial context, it is natural to attempt to promote the discrete domain to a location parameter, which therefore has a continuous character, whether in its full three-dimensionality or in simplifications to one or two dimensions. We focus on two dimensions, for intended applications to histopathology.

For a continuous probability distribution given by a density $p(x,y)dx\thinspace dy$ on a domain $D\subset \mathbb{R}^{2}$, the \emph{differential entropy} can be defined with respect to a reference density $q(x,y)dx\thinspace dy$:
\begin{align}
 H_{DE} := \int_{D} p(x, y)\left(-\operatorname{log}(p(x,y)/q(x,y))\right)dx\thinspace dy
\end{align}
In the literature there are claims that differential entropy, as traditionally defined in \cite{shannon1948mathematical} with $q(x,y)$ constant, is an inadequate generalization of ordinary discrete entropy. For example Jaynes \cite[p201]{jaynes1957information} claims that Shannon's differential entropy was defined ``without calculating''. But from \cite[p37]{shannon1948mathematical} it seems that Shannon was well aware of the effect of coordinate transformation and the implicit dependence on a reference volume form. Although it is not as rigorous as \cite[Theorem 3]{shannon1948mathematical}), in \cite[p38]{shannon1948mathematical} we find a statement analogous to Theorem 3 which if true substantially justifies this concept of differential entropy. Jaynes' proposed variant is essentially the same except for the explicit attention on the reference measure. The real criticism seems directed at the failure of $H_{DE}$ to be expressible in a natural way as a limit of ordinary discrete entropies (for one example of such criticism, see \cite[p4]{batty1974spatial}).

We suggest that this is simply a fact about $H$ and $H_{DE}$ and not a defect. In fact, we can use this ``defect'' to our advantage by considering finite expressions of $H$ which are putative ``approximations" of $H_{DE}$ as recording non-trivial features of the approximation scheme or triangular mesh rather than merely approaching $H_{DE}$ in some limit. That is, $H=H(m, p)$ is a function of the mesh $m$, which can be used to study $m$ itself.

To clarify this situation, contrast it with integration. There are a wide variety of discretization schemes $R=R(m_{\varepsilon})$ for Riemann sum approximations of integrals that all result in the same value as the mesh size parameter $\varepsilon$ limits to 0, the integral. Consequently Riemann sums are not used to distinguish or study partitions of the interval $[0,1]$, for example. But the various partitions of the uniform interval $[0,1]$ into $N$ parts results in every possible finite distribution $(p_1, ..., p_N)$, and hence results in every possible entropy value, and so we do not expect ``independence of mesh'' for entropy, and differential entropy should be conceptually differentiated from finite entropy expressions.

In \cref{sec:pointprocess}, we argue that a continuous distribution assumption is often incorrect for point pattern analysis, thus we will not need recourse to continuous distributions.

\paragraph{Basins, Voronoi/Thiessen Polygons, Delaunay Triangles.}
In the direction of attending to the mesh itself, we consider the Voronoi polygon division associated to a given point set $\{v_i\in \mathbb{R}^{2}\}$ in the plane. Invoking a dynamical interpretation, we call these \emph{basins}, as in basins of attraction to a central point. For the computation of the Voronoi tiling and Delaunay triangulation we use the {\tt scipy} Python library \cite{2020SciPy-NMeth} as an interface to QHull \cite{barber1996quickhull}.
 
\paragraph{Uniform Basin and Dense Basin Distributions and Their Entropy.}
In our initial investigation, we sought to use the basin tiling as a density estimator for a hypothetical continuous distribution from which sample sets are drawn in applications like digital histopathology of biological cell sets. The goal was to compute the differential entropy from samples in a manner distinct from Kozachenko and Leonenko's graph-based approach \cite{kozachenko1987sample} by relying more on local geometry. This would be in line with the Voronoi density estimators presented in \cite{polianskii2022voronoi} and \cite{marchetti2023efficient}. 

Our previous remark about entropy of partitions of the uniform interval $[0,1]$ suggests that a simpler construction may be worthwhile: a Voronoi-cell-based discrete approximation of the \emph{uniform} area distribution in the plane. This discrete distribution we call the \emph{uniform basin distribution}: $p^{\text{ub}}_i = A_i / A$, where $A_i$ is the area of basin $i$ and $A=\sum_{i}A_i$ (all inifinite regions omitted). This distribution approximates the symbol distribution of the discrete process which is induced by uniform Brownian motion (also known as the Wiener stochastic process) in the relevant planar domain by recording basin membership transitions in trajectories or realizations of the process. Despite the phenomenon we observed above in the the analogous 1-dimensional case, we found in preliminary simulations that $H_{\text{ub}}$ (the entropy of $p^{\text{ub}}$) is not a strong discriminator of non-randomness in planar point patterns. After all, the uniform distribution itself is not carrying any information about the point set, so the fact that $H_{\text{ub}}$ is not an effective discriminator is in a sense the expected result.

We found the \emph{dense basin entropy} $H_{\text{db}}$ to be more informative. It is defined as the entropy of the \emph{dense basin distribution}: $p^{\text{db}}_i = (1/A_i)/Z$, where $Z=\sum_i 1/A_i$. This distribution, consisting of ``dense basins'', with 1 point per basin but spread out over its area, can be regarded as a discretization of a hypothetical continuous density of which the given point set is a sample. This is the same as the distribution used by Ord \cite{ordtree} to estimate the Poisson process density $\lambda$.

\subsection{Random Point Processes and Testing for Patterns}
\paragraph{Poisson Process and Alternative Null References.}
\label{sec:pointprocess}
The Poisson point process is the formal model in stochastic process theory which describes point sets in which each point is drawn independently of the others with respect to the standard uniform measure, idealized over the whole infinite plane. The Poisson process and finite area variants are often the implicit default hypothetical reference model for point pattern analysis.

However, there are a number of problems with applying this concept to real empirical data. One issue is ``the validity and meaning of treating an observed population as though it is a sample from a hypothetical population'', as mentioned in prior work \cite{summerfield1983populations}. This is indeed problematic for cell microscopy and numerous other application domains, where there is often no concept of drawing ever more samples in a given study area because the population is finite and the study is already exhaustive; to this end, Upton and Fingleton aptly remark that ``Great imagination has gone into turning what appears to be a population into a sample'' \cite[p325]{upton1985spatial}. Indeed the reason for finiteness in such cases is often because the points are simplified, extent-free idealizations of the location of objects which in reality possess a definite extent. This causes the point sets to exhibit systematic location-based exclusion or repulsion that is inconsistent with the Poisson process model. Several prior studies have made similar observations (e.g. ``minimum interpoint distance'' constraint in \cite[p405]{okabe1999spatial}, the ``hard discs model'' of \cite[p182]{vincent1982theoretical}). The model known as Simple Sequential Inhibition (SSI) assumes a minimum interpoint distance and fixed number of points.

Intuitively, when point set sizes and expected object diameters are small enough relative to study area, alternative null reference models with either finiteness constraints or exhibiting local-exclusion dependence may be expected to approximate closely to the Poisson regime, and this is confirmed in simulations by Vincent and Haworth \cite[184]{vincent1982theoretical}. Nevertheless investigators should take steps to mitigate problems of interpretation in cases (not at all rare) where the Poisson regime is not in effect. The classical results for Delaunay triangle statistics under the Poisson assumption are reviewed in \Cref{delaunaynull}.

\paragraph{Lower bounds on $H_{\text{db}}$.}
A clear lower bound for the dense basin entropy was observed in experiments evaluating this entropy on cell sets from histopathology datasets. Knowing that such point sets exhibit local exclusion as in the SSI assumption, with minimum interpoint distance owing to known sizes for biological cells, we hypothesized that this entropy is bounded below in terms of the minimum distance and the number of points. We proved a lower bound which appears to be relatively weak. We also found that an alternative heuristic largely exhibits the behavior observed in empirical experiments.

Let $r>0$ and set $a:=\pi r^2$, the area of a disc with radius $r$. Let $V_a$ be the set of configurations $\{v_i\}$ in the plane with $N$ finite Voronoi cells such that $|v_i- v_j| > 2r$ for all $i,j$. $A_i$ is the area of Voronoi cell containing $v_i$, so we have $A_i > a$. Consider $A:=\sum_{i=1}^{i=N}A_i$ as a function on $V_a$. Let $p_i:=(1/A_i)/Z$ be the Ord distribution. Recall that $F(x):=x\cdot (- \operatorname{log}x)$.

\begin{proposition}\label{basicineq} Assume $N>2$. Let $c:=\left(\frac{a}{A}\right)^{N}(N-1)!\thinspace\thinspace$.
  \vspace{-0.5pc}
  \begin{enumerate}
    \setlength\itemsep{-0.3em}
    \item{$p_i > c$ for all $i$.}
    \item{On $V_a$ the dense basin entropy satisfies $H_{\operatorname{db}}\geq N\cdot F_{\operatorname{min}}$, where \[F_{\operatorname{min}}:=\operatorname{min}\left\{ F(c), F(1-(N-1)c) \right\}\]}
  \end{enumerate}
\end{proposition}
\vspace{-0.75pc}
\emph{Proof}. (1)
\[ \frac{1}{Z} = \frac{1}{\sum 1/A_i} = \frac{\prod_i A_i}{\sum_i \prod_{j\neq i}A_j}=: \frac{P}{Q}\]
$P > a^{N}$ because $A_i > a$. $Q$ is a sum of terms in the expansion of $A^{N-1}=(A_1+\cdots + A_N)^{N-1}$. Each contributing term occurs $(N-1)!$ times due to reordering, so $Q < (1/(N-1)!) A^{N-1}$. Together with $1/A_i > 1/A$, we have
\[ p_i = (1/A_i)/Z = (1/A_i)\cdot P \cdot Q^{-1} > \frac{1}{A} \cdot a^{N} \cdot (N-1)! \frac{1}{A^{N-1}}=c \]

(2) Now since $p_i > c$ for all $i$, we also have $p_i = 1 - \sum_{j\neq i}p_j < 1 - (N-1)c$. $F(x)$ is a concave function and when restricted to $[c, 1-(N-1)c]$ its minimum is $F_{\text{min}}$ as defined above. Thus $H_{\text{db}}=\sum_{i=1}^{i=N}F(p_i)>N \cdot F_{\text{min}}$. \qedsymbol{}

The bound obtained by the simpler argument below is a close match to the empirical observation, although it is only a heuristic local bound, holding for the idealized dense basin distribution of greatest concentration at one point.

\begin{proposition}\label{idealHdb} Let $A_i>0$ for $i=1, ... , N$ satisfy $A_i \geq a > 0$. Let $A=\sum_i A_i$. Let $p_i= (1/A_i) / Z$, where $Z=\sum_i 1/A_i$. Under these conditions the value $\bar{p}$ of $p$ for which $p_1$ is as large as possible and $p_{i\neq 1}$ are all equal is:
  \vspace{-0.75pc}
  \[ \bar{p}_1 = (1/a)/Z \qquad \bar{p}_{i\neq 1}= \tfrac{N-1}{A-a}\cdot \tfrac{1}{Z} \qquad Z=\tfrac{1}{a} + \left(\tfrac{N-1}{A-a}\right)(N-1)  \]
  For this $\bar{p}$:
  \vspace{-0.25pc}
  \[ H(\bar{p}) = F(\bar{p}_1) + (N-1) F(\bar{p}_{i\neq 1})  \]
\end{proposition}
\emph{Proof}. The formula for $\bar{p}_1$ is directly due to $A_1 \geq a$. Then $\sum_i \bar{p}_i=1$ implies the formula for $\bar{p}_{i\neq 1}$. \qedsymbol{}

Note that Proposition \ref{idealHdb} does not claim that $\bar{p}$ is achievable as a dense basin distribution for an actual Voronoi tiling, and it also does not claim that realized values of $H_{\text{db}}$ must satisfy $H_{\text{db}} > H(\bar{p})$, even though this was often observed.

\subsection{Voronoi-Markov Process}

\paragraph{Transition Probabilities, Entropy Rate, and Spectral Gap.}
The dense basin entropy as well as the Delaunay statistics have the deficiency that they only capture spatial information by means of the attribute distribution for the individual basin polygons. In principle the same distribution could be possible for very different arrangements. In particular, neighbor connectivity of basins is not used.

We therefore propose a finite Markov chain adapted to the basin structure making direct use of neighbor connectivity. By modeling this chain loosely on a Brownian motion-like stochastic process on the underlying plane, it also has the advantage that a successful analysis with this model may be interpretable in the mechanistic/dynamical terms of the given application domain, e.g. objects traversing the domain visiting the basins.

For the set of states we use the given planar point set, $\{v_i\}$. For fixed $i$, the transition probabilities $p_{ij}$ to the various $j$ are chosen to be proportional to the area ``near $i$ but between $i$ and $j$'' when the basins of $i$ and $j$ share an edge (known as a ridge in {\tt scipy}) of length $r$. We take this area to be the triangle spanned by $v_i$ and the segment of length $r$. If $u$ is the distance $|v_i - v_j|$, then this triangle has area $r u / 4$, because the ridge is perpendicular to the segment joining $v_i$ and $v_j$ (see Figure \ref{fig:vm}a). Finally, as an optional regularization, we choose a constant reference area $a$, considered perhaps as a disc surrounding $v_i$, and allow transition probability from $i$ to $i$ proportional to $a$. In our simulations we suppress hyperparameters by setting $a\approx 0$.

By proportional we just mean that the rows $p_{ij}$ are normalized to have sum 1, making it a right-stochastic matrix. That is, more explicitly we define the Voronoi-Markov chain by:
\begin{align}
p_{ij}^{VM} = \begin{cases}
  \frac{1}{4} ru/z_{i} & v_i, v_j \text{ are neighbors, } |v_i - v_j|=u, \text{ boundary segment } r \\
  a/z_{i} & i=j
\end{cases}
\end{align}
where $z_i=a + \sum_{\text{neighbors }j}\tfrac{ru}{4}$. 
For a typical basin $i$, not adjacent to infinite cells, all of the area $A_i$ within the basin will appear in the sum, which is $z_i=A_i + a$ in this case. Thus along the diagonal, $p_{ii}=a/(A_i + a)$ is not constant (as it might appear from our pre-normalized definition invoking $a$) but instead is very similar to the dense basin distribution.

This choice of Markov process can be understood as a finite-support process approximately induced by planar Brownian motion, or the Wiener process. Brownian motion is associated with the uniform distribution (the result of diffusion), and the idea is that the frequency of boundary-crossing events should be proportional to the mass of particles which may reach the given boundary segment. 

In simulations we found that the invariant measure $\mu$ of $p^{VM}$ tends to behave as the opposite of the point set density (see \Cref{fig:vmoffset}), so $1/\mu_i$ plays the role of the point density. We also considered a more complicated variant of $p^{VM}$ in which transition probability is increased up-gradient for the point density, in order to achieve consistency between the invariant measure and the desired density. But the simpler model defined above seems adequate, provided that it is suitably interpreted.

Note that the passage from the Wiener process to a discrete process by associating trajectories to their nearest neighbor basin centers would not result in a process with the Markov property, i.e. memorylessness, as trajectories meeting a given polygonal edge are nearer to some edges than to others in the new basin. That is, the probability of transit to another basin does not depend only on being located within a given basin, but also on where it entered. Therefore the Voronoi-Markov chain can only be an approximation of such an induced process.

There is associated to any finite Markov chain its \emph{entropy rate} (see \cite[p77]{cover1999elements}), as well as its \emph{spectral gap}. The spectral gap is the difference between the first two eigenvalues, ordered by modulus. Since 1 is always an eigenvalue (for the invariant measure), and indeed 1 is the largest-modulus eigenvalue, the spectral gap is (the absolute value of one minus) the next-largest eigenvalue $\lambda_2$. $\lambda_2$ is an important characteristic that can be shown to control the rate of convergence of trajectories' distribution to the invariant measure. The spectral gap is important in a number of mathematical and physical fields.

\paragraph{Invariant Measure For Two-Set Comparison.}
The Voronoi-Markov construction provides a mechanism for comparison of two point sets $\{v_i^1\}$ and $\{v_j^{2}\}$, which is useful for point set pair pattern analysis. The Voronoi-Markov construction is made for the $v_i^1$, and the points $v_j^2$ are assigned to the $v^1$ basins in which they are contained, resulting in a frequency distribution $\nu_i$. If $\mu_i$ denotes the Markov chain invariant measure for points $i$ of the first set, define the \emph{VM concordance}: $1/\sum_i \mu_i \nu_i$. This measures the degree to which the points of the second set are present along the highest-probability areas for trajectories of the Markov chain, and so it may be regarded as an assessment of whether or not the points of the second set might be a snapshot of dynamics that are a match for the predictions of the Markov chain. Following our remark about the $p^{VM}$ invariant measure, we also assess the \emph{inverted VM concordance}: $\sum_i \nu_i / \mu_i$.

\paragraph{3-point model.}
The simplest non-trivial case of the Voronoi-Markov chain has 3 states and 3 basin generating points, provided we make a minor modification to the construction. This case will never be directly useful for applications, but it may provide hints about the mathematical character of the general case, and it has the virtue of tractability (see \Cref{threepoint}).

\section{Results}

\paragraph{Concentrated Disc Model Simulation.} To demonstrate the ability of the dense basin entropy, Voronoi-Markov quantities, and the Delaunay shape statistics to discern patterns in random point sets, we use a \emph{concentrated disc model}. In this model the standard uniform distribution on the unit disc in the plane is altered by the radial transformation $r \mapsto r \cdot r^{q}$, where $r$ is the distance from the origin and $q>0$ is a concentration parameter. When $q=0$ no alteration is made, and larger $q$ increases the density in the vicinity of the origin. A number of random samples of $N$ points are drawn for various $N$. Examples of the point sets used are shown in \Cref{fig:concentrateddisccells}. The average value of each metric over 100 repetitions of each case is plotted in \Cref{fig:concentrateddiscplots}. 

The normalized dense basin entropy $H_{\text{db}}/\log_2 N$ detects concentration well, roughly halving between $q=0$ and $q=2$, with greater efficacy for larger $N$. The (scaled) Voronoi-Markov spectral gap exhibits a sharp sigmoid-like cutoff at a non-trivial concentration around $q=1$, especially for larger $N$. Its entropy rate shows a more stable, linear response to $q$. The Delaunay (general) angles and the Delaunay minimum angles tested in distribution against the null distribution shows a clear trend towards significant $p$-values with increasing $q$, albeit with apparent noise.

\paragraph{Concentrated Disc Model, Two-Set Comparison.} The Voronoi-Markov invariant measure is used for pattern alignment detection in a two-set model. Both point sets are drawn from the concentrated disc but with the second set offset to the right by an amount $d$. \Cref{fig:heatmaps} shows heatmaps of the VM concordance and inverted VM concordance metrics for various concentrations and offsets. For a given concentration, one hopes to detect small offsets as a signal of concordance between the point sets, with finer detection needed for higher concentrations.

Surprisingly both metrics are effective detectors, though in qualitatively different ways. They exhibit respectively high sensitivity versus high specificity. Without concentration, i.e. for $q=0$, the point sets are uniform and there is \emph{a priori} no basis for detection of the offset displacement between them, so these two metrics represent the two extremes for assignment in the ambiguous cases.

The VM concordance is broadly effective when $q>2$, predicting concordance with perhaps some false positives for relatively larger offsets. For $q<2$ the metric is uniformly high, so it is uninformative for low concentration scenarios. The inverted concordance metric has much sharper response, detecting concordance only for very small offsets when $q$ is high enough (also around $q=2$), with a very low chance of false positives.

\paragraph{Evaluating spatial entropies and other metrics on lung cancer and bone marrow cell set patterns.}\label{compexp} The proposed spatial entropies, the Voronoi-Markov-chain-derived quantities, and the Delaunay angle statistics were computed on cell subsets in the 536 lung cancer samples of Sorin et al. \cite{sorin2023single}, and the 30 bone marrow samples of Sarachakov et al. \cite{sarachakov2023spatial}. The empirical investigation of these non-trivial point patterns shows that each of the metrics offers independent information. Their distribution is plotted in \Cref{fig:featuregrid}.

Across 20 cases, in the 5 cases achieving $t$-test $p$-value less than $0.10$ the metrics appearing were the VM spectral gap $\sigma$, Delaunay $\alpha_{min}$ test $p$, and $H_{\text{db}}$. The same metrics appear with respect to achievement of AUC values greater than $0.65$. Top performers in the best cases were the spectral gap (AUC=$0.96$, $p=9\cdot 10^{-4}$), dense basin entropy (AUC=$0.95$, $p=0.05$), and Delaunay $\alpha_{min}$ (AUC=$0.93$, $p=3\cdot 10^{-3}$). The uniform basin entropy never achieved AUC more than $0.64$ or $p$ less than $0.29$.

\section{Limitations and Future Work}\label{limitationssection}
Entropy can be formulated in the spatial context and effectively applied to spatial point pattern detection and stratification by relying on the Voronoi discretization of the spatial domain that is adapted to observed point sets. A promotion of the construction to a Markov chain accounting for additional local geometry, in the basin adjacency structure, also yields effective discrimination metrics, especially the spectral gap. These measures compare favorably with classical Poisson-Delaunay statistics based on analytical null distributions.

However, we have argued in favor of null references, and the new metrics promoted here (the dense basin entropy, Voronoi-Markov entropy rate and spectral gap) would be much more useful in practical applications if their distribution were established precisely under the conditions of realistic null point process models. Further empirical studies are also needed to evaluate the real utility of the proposed metrics in detecting signal and stratifying patterns in a variety of datasets. There may also be viable alternative definitions of the Voronoi-Markov chain, especially if closer alignment with Brownian motion is desired. The computation of the VM spectral gap requires eigenanalysis for a square $N \times N$ matrix for $N$ points, but since this matrix is sparse, practical computational efficiency beyond $N\approx1000$ will be possible using an implementation that is effective in the sparse regime.

Preliminary investigation (see \Cref{eigensec}) of the second and third eigenfunctions for the Voronoi-Markov chain suggest that they are analogous to the principal components in Principal Component Analysis (PCA), or to the low-frequency modes in harmonic analysis. Indeed, this Markov chain is closely related to an edge-weighted Delaunay graph, and a possible relation with the graph Laplacian should be explored, especially if a construction can be found which shows that making use of the 2-dimensional cells may provide a strictly richer structure than the 1-dimensional graph. 

\begin{figure}[tbh]
  \centering
  \begin{subfigure}{0.39\linewidth}
    \includegraphics[width=6cm]{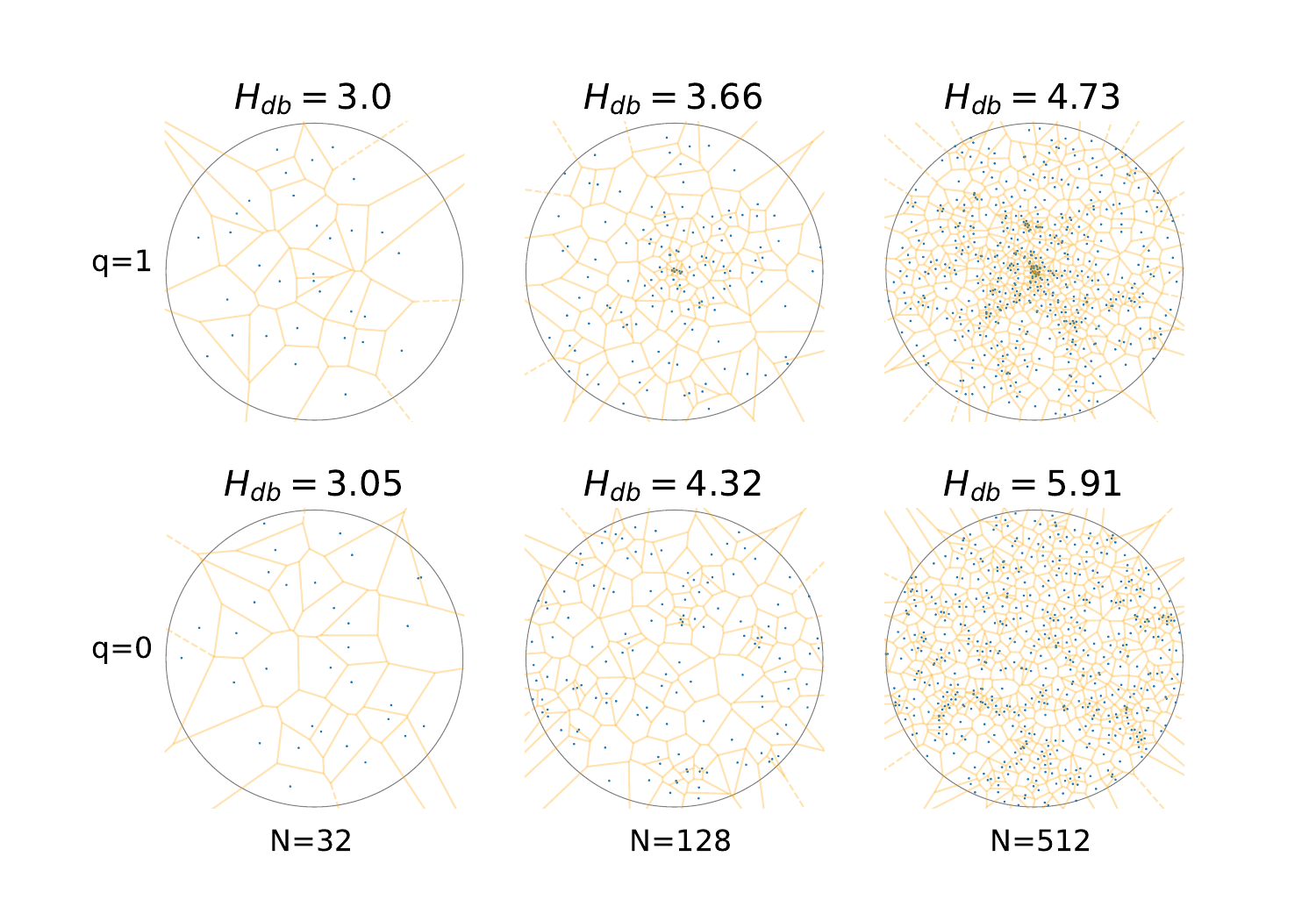}
    \caption{\label{fig:concentrateddisccells} Voronoi cell structures for $N$ points drawn from the concentrated disc, the unit disc with the uniform distribution transformed by $(x,y)\mapsto (x,y)\cdot r^{q}$, where $r=|(x,y)|$. Since $r^{q}\leq 1$ for $q>0$, density is concentrated near the origin as $q$ increases. The dense basin entropy $H_{db}$ is shown for each case.}
  \end{subfigure}\hspace{0.4pc}
  \begin{subfigure}{0.59\linewidth}
    \includegraphics[width=8.3cm]{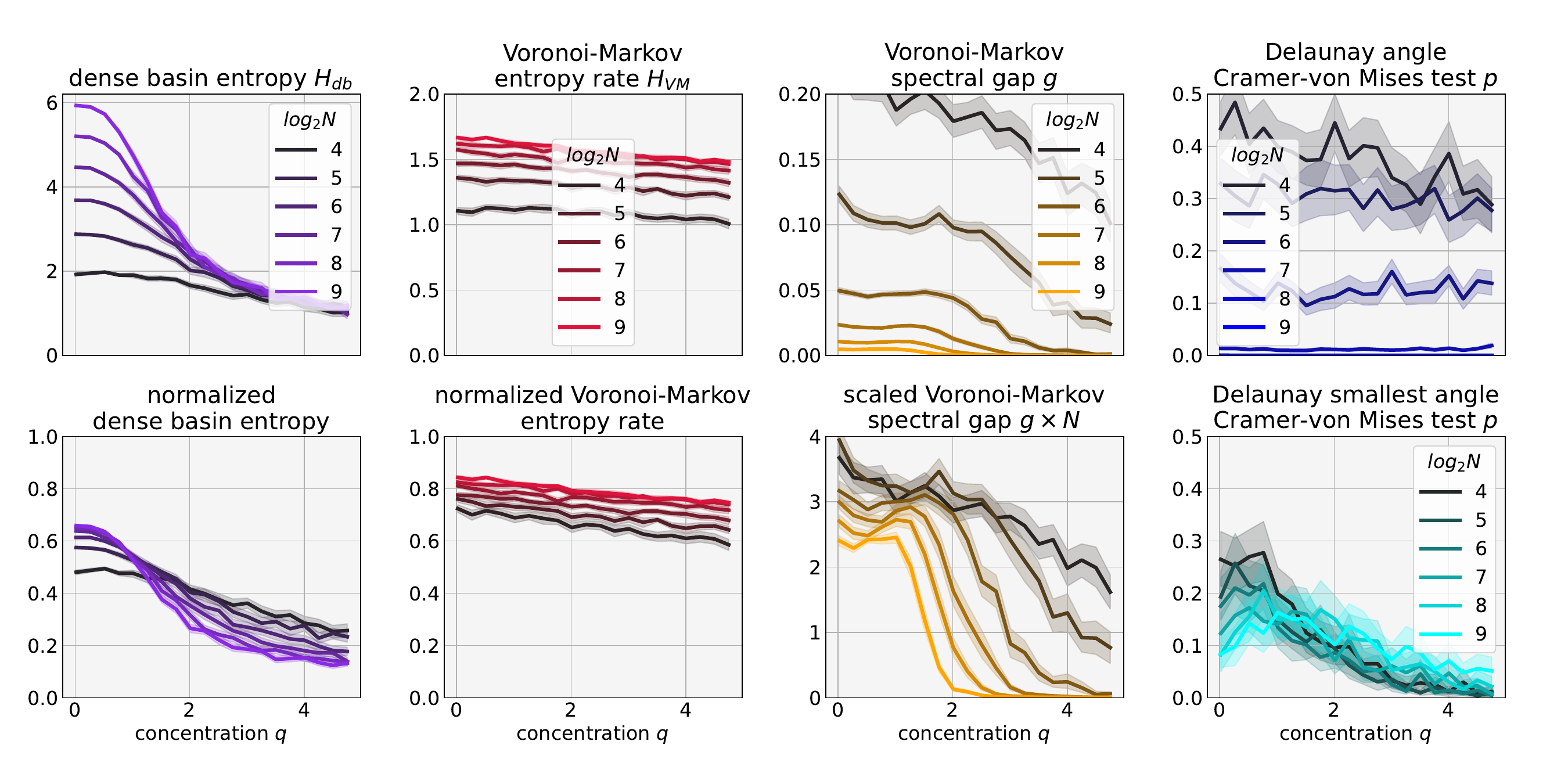}
    \caption{\label{fig:concentrateddiscplots} The dense basin entropy, Voronoi-Markov entropy rate and its spectral gap, with normalized variants, and Delaunay angle and minimum angle distribution goodness-of-fit test $p$-values, for points sets drawn from the concentrated disc. 100 point sets of size $N$ were drawn for each value of the concentration parameter $q$, and the plots show the average value of each metric.}
  \end{subfigure}
  \caption{\label{fig:concentrateddisc} The concentrated disc model for various degrees of concentration $q$, evaluated with the dense basin entropy, metrics derived from the Voronoi-Markov chain, and Delaunay triangulation angle statistics.}
\end{figure}

\begin{figure}[t!]
\centering
\begin{subfigure}{0.49\linewidth}
\includegraphics[width=7.4cm]{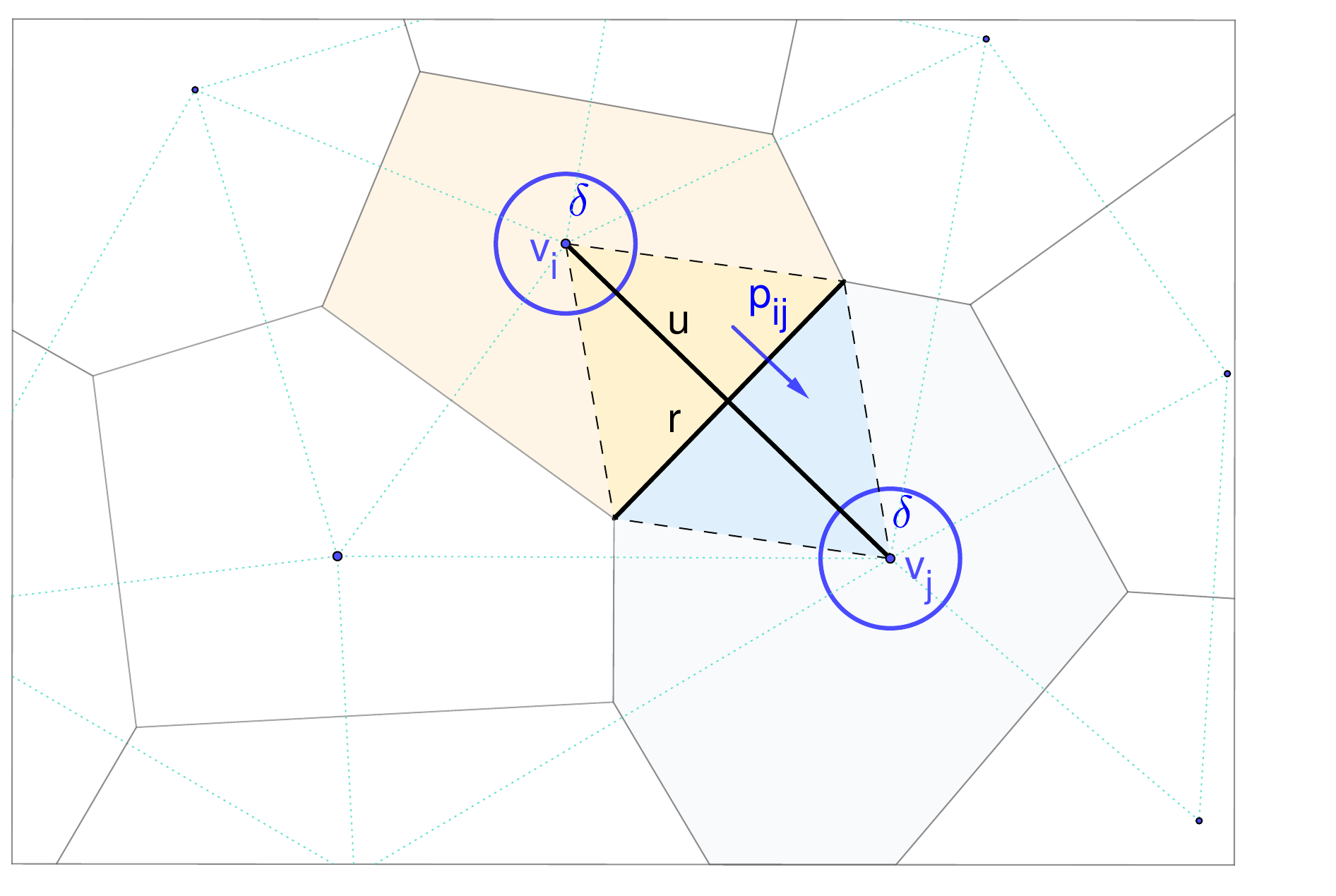}
\caption{\label{fig:markovcell} Illustration of the definition of the transition probabilities for the Voronoi-Markov chain. Transit frequency from $i$ to $j$ is proportional to the triangular area bounded by transition segment $r$ and the central vertex $v_i$.}
\end{subfigure}\hspace{0.25pc}
\begin{subfigure}{0.49\linewidth}
    \includegraphics[width=8cm]{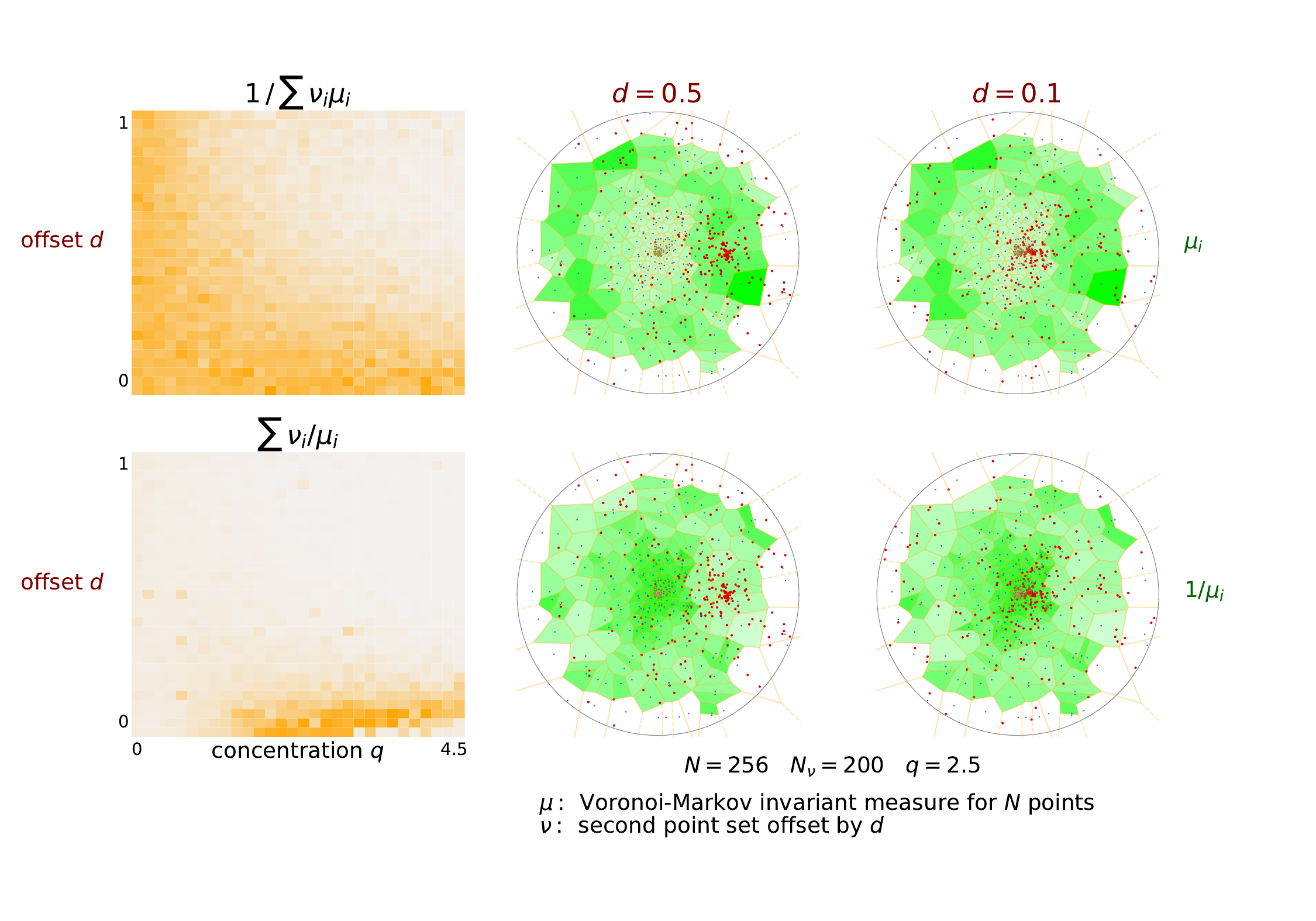}
    \caption{\label{fig:heatmaps}Two concentrated-disc point sets, offset by an amount $d$, are evaluated with the VM concordance and inverted concordance metrics (represented in orange) derived from the invariant measure $\mu$ of the Voronoi-Markov chain of the first set, and the nearest-neighbor assignment binned distribution $\nu$ for the second set. $\mu$ (above) and $1/\mu$ (below) are shown in green, and the points of the set forming $\nu$ are marked red.}
\end{subfigure}
  \begin{subfigure}{0.99\linewidth}
  \includegraphics[width=0.88\textwidth]{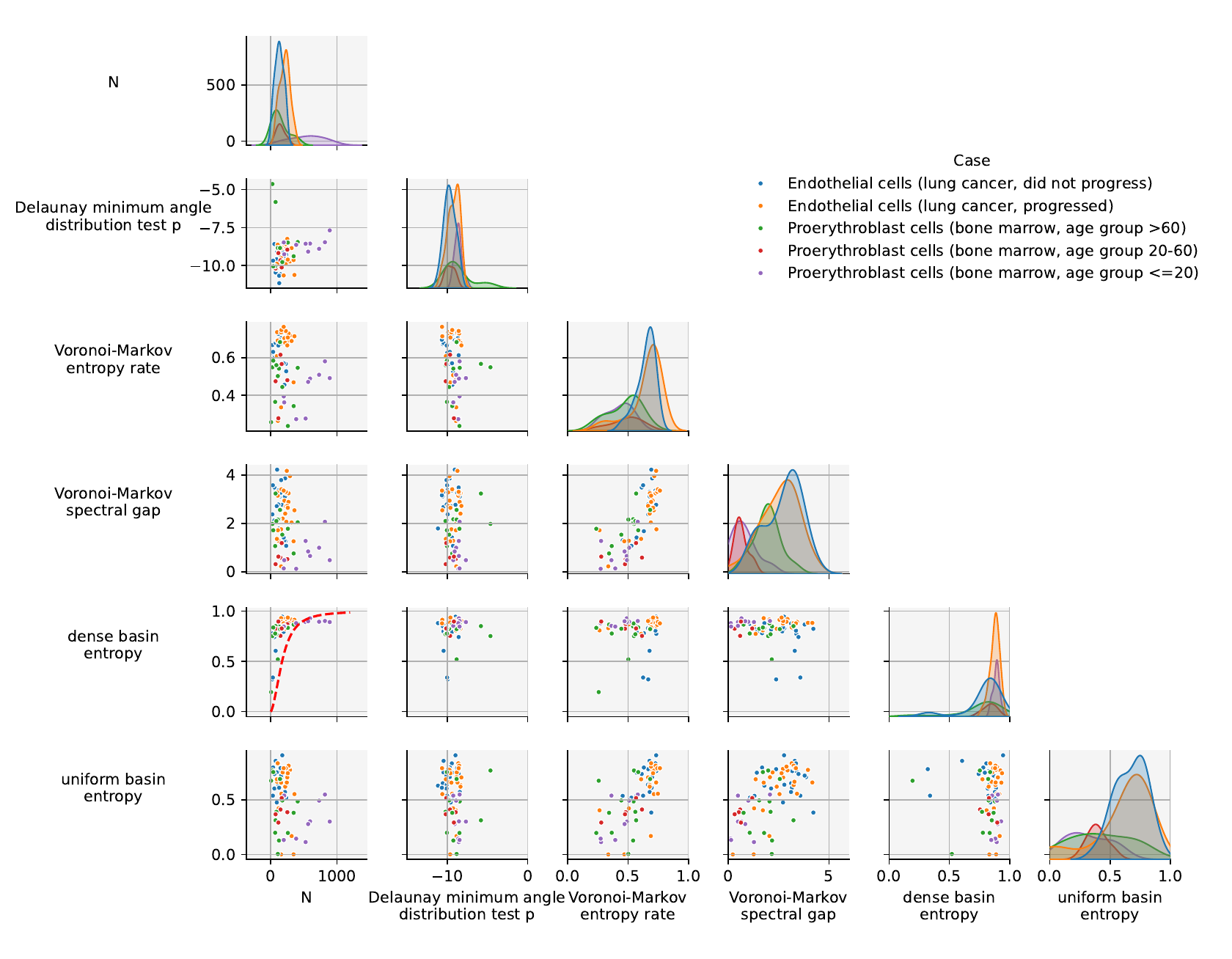}
    \caption{\label{fig:featuregrid} The proposed spatial entropies and other metrics were evaluated on certain cell phenotype subsets in the lung cancer samples of Sorin et al. \cite{sorin2023single} and the bone marrows of Sarachakov et al. \cite{sarachakov2023spatial}. Lung cancer cohorts were downsampled to 50 samples for legibility. Shown are the KDE plots on the diagonal, and the joint distribution of each pair of metrics. The number $N$ of cells is also included, for comparison as a potential confounder. Only normalized variants of the metrics are used (and the logarithm applied to the Delaunay test $p$ values). Most metrics are generally independent of $N$, except for the dense basin entropy, which exhibits a lower bound. The heuristic lower bound of Proposition \ref{idealHdb} is shown in dashed red with parameter $A/a=50000$. The Markov entropy rate and spectral gap show clear correlation. Several of the joint distribution plots indicate strong independence, meaning that they measure different aspects, which is somewhat surprising given that each of the included metrics was a reasonable predictor of concentration in the concentrated disc model. The spectral gap appears to be effective at discriminating proerythroblast cell sets between young and old bone marrows.}
  \end{subfigure}
    \caption{\label{fig:vmoffset} The Voronoi-Markov chain definition, two-set joint analysis, and metrics evaluated and compared for lung cancer and bone marrow histopathology cell sets.}
\label{fig:vm}
\end{figure}

\clearpage

\textbf{Acknowledgments.} This project was supported by NIH R37CA295658, MSK Technology Development Fund, and MSK Cancer Center Support Grant/Core Grant (P30 CA008748).

\printbibliography

\appendix

\section{Reference Null Distributions of Delaunay Triangle Attributes}\label{delaunaynull}

Boots et al. \cite{boots1986using} carried out empirical studies comparing the effectiveness of goodness-of-fit tests of the Delaunay angle-based statistics for discriminating point patterns, and found that the Cramer-von Mises and Anderson-Darling tests were the most effective (the Kolmogorov-Smirnov test, popular perhaps for its theoretical importance, was the least effective). They also compared the specific angle used. Distributional formulas are available for the average $\alpha$, the minimum $\alpha_1$, intermediate $\alpha_2$, and maximum angle $\alpha_3$. Boots et al. considered $\alpha$, $\alpha_1$, and $\alpha_3$, and their findings suggest that they all represent significant metrics with complementary (rather than overlapping) information.

\begin{proposition} Consider as random variables the quantitative attributes of the Delaunay triangulation on Poisson process point sets in the plane. The angle $\alpha$ of a triangle and the smallest angle $\alpha_1$ are distributed as $f(\alpha)d\alpha$ and $f_1(\alpha)d\alpha$ respectively with density functions:
  \begin{align}
    f(\alpha)&=\frac{4}{3\pi}\left[(\pi - \alpha)\cos \alpha + \sin \alpha\right](\sin \alpha) \quad& \alpha\in(0,\pi) \\
    f_1(\alpha)&=\frac{2}{\pi}\left[(\pi - 3\alpha)\sin 2\alpha + \cos 2\alpha -\cos 4\alpha\right] \quad& \alpha\in(0, \pi/3)
  \end{align}
\end{proposition}

$f(\alpha)$ is the original Collins-Miles formula \cite{collins1968geometrical, miles1970homogeneous}. $f_1(\alpha)$ is reported in \cite{boots1986using}, citing \cite{puri1978analysis}. There are also formulas for $f_{2}(\alpha)$ and $f_{3}(\alpha)$ (slightly more complicated), and the joint distribution of two angles $(\alpha, \beta)$ of such triangles. Many other properties are also well understood like the circumradius and side length distribution. Since $(\alpha, \beta)$ forms a complete set of shape parameters for a triangle up to scale, an effective test devised against the known distribution of $(\alpha, \beta)$ would be particularly useful for point pattern analysis, as it would comprise an assessment of the full shape content in the Delaunay triangles of empirical point sets.

Integration of $f(\alpha)$ and $f_1(\alpha)$ provides the cumulative distribution function for $\alpha$ as required by goodness-of-fit tests:
\begin{align}
CDF(\alpha)&=\int_{0}^{\alpha}f(x)dx = \tfrac{1}{6\pi}\left[ 4\alpha + 2\pi -3\sin 2\alpha - 2(\pi - \alpha)\cos 2\alpha \right] \\
  \begin{split}
  CDF(\alpha_1)&=\int_{0}^{\alpha_1}f_1(x)dx \\
  &= \tfrac{1}{\pi}\left[ 3\alpha_1 \cos 2\alpha_1 -\sin \alpha_1 ( 2 \cos \alpha_1 + \cos 3\alpha_1 -2\pi \sin\alpha_1 ) \right]
  \end{split}
\end{align}

The Cramer-von Mises and Kolmogorov-Smirnov tests for custom CDFs are implemented in {\tt scipy}, making the ``Delaunay shape statistics'' tests for $\alpha$ and $\alpha_1$ based on these CDFs readily computable.

Although the derivation is not difficult, we were not able to find formulas for these CDFs in the literature.

\section{3-point model}\label{threepoint}
The simplest non-trivial case of the Voronoi-Markov chain has 3 states and 3 basin generating points, provided we make a minor modification to the construction. This case will never be directly useful for applications, but it may provide hints about the mathematical character of the general case, and it has the virtue of tractability.

The basins for a 3 point set are all infinite, so we restrict the areas considered in the construction to the portions lying within the single Delaunay triangle. Practically this means that the area expression $ru/4$ for neighbors at a distance $u$ will be replaced by $hu/4$, where $h$ is the portion of the perpendicular bisector of segment $u$ which lies in the triangle. A generic depiction of one Delaunay triangle in a Voronoi cell structure is shown in \Cref{fig:trianglecell}, to illustrate these areas.

Let $u_1$, $u_2$, $u_3$ be the interpoint distances (Delaunay triangle side lengths), and $a_i=h_i u_i/4$. We first compute the dense basin distribution, providing a direct formula for the dense basin entropy in this special case.

\begin{proposition} For the 3-point model:
  \begin{enumerate}
    \item{The uniform basin distribution is $\frac{1}{A}\cdot\left( a_2 + a_3, a_1 + a_3, a_1 + a_2 \right)$, where $a_i=h_i u_i / 4$, $h_i=\sqrt{R^{2}-u_i^2}$ for the Delaunay triangle circumradius $R$, and $A=2(a_1 + a_2 + a_3)$ is the triangle area.
    }
    \item{The dense basin distribution is
      \begin{align}
        \frac{1}{p_2 + 3 h_2}\cdot \left((a_1 + a_2)(a_1 + a_3), (a_2+a_1)(a_2+a_3), (a_3+a_1)(a_3+a_2) \right)
      \end{align}
      where $p_2=a_1^2 + a_2^2 + a_3^2$ denotes the power sum and $h_2=a_1a_2 + a_1a_3 + a_2a_3$ denotes the homogeneous symmetric polynomial.
    }
  \end{enumerate}
\end{proposition}
The Voronoi-Markov chain $(p_{ij})$ has the form of the right-stochastic matrix
\begin{align}
 \begin{bmatrix} a/s_1 & a_3/s_1 & a_2/s_1 \\ a_3/s_2 & a/s_2 & a_1/s_2 \\ a_2/s_3 & a_1/s_3 & a/s_3 \end{bmatrix}
\end{align}
where $s_1=(a_2+a_3+a)$, $s_2=(a_1+a_3+a)$, $s_3=(a_1 + a_2 + a)$. Recall that $a$ is a fixed constant area.

The invariant measure and the spectral gap under some conditions can be determined by direct methods.
\begin{proposition}For the 3-point Voronoi-Markov chain, assume all $a_i$ are distinct.
  \begin{enumerate}
    \item{\label{invm}The invariant measure of $p$ is proportional to $(s_1, s_2, s_3)$.}
    \item{\label{specM}The eigenvalue spectrum of $p$ equals $\{1, \lambda_2, \lambda_3\}$ with either $\lambda_3=\overline{\lambda_2}$ or both $\lambda_2,\lambda_3\in \mathbb{R}$. $\lambda_2$ and $\lambda_3$ are the solutions of
      \begin{align}
        0 = x^2 +\left(1-a\left(\frac{1}{s_1}+\frac{1}{s_2}+\frac{1}{s_3}\right)\right) x + \frac{2 a_1 a_2 a_3 - (a_1^2 + a_2^2 + a_3^2) a + a^3}{s_1 s_2 s_3}
      \end{align}
    }
    \item{\label{l2}In the case $(\lambda_2, \lambda_3)=(\lambda_2, \overline{\lambda_2})$, the (absolute) spectral gap of $M$ is
      \begin{align}
        g = 1-{\frac{\sqrt{2 a_1a_2a_3 -(a_1^2 + a_2^2 + a_3^2) a + a^3}}{\sqrt{s_1 s_2 s_3}}}
      \end{align}
    }
  \end{enumerate}
\end{proposition}
\emph{Proof}. (\ref{invm}) This can be verified by direct computation, for example \begin{align}
  s_1(a/s_1)+s_2(a_3/s_2)+s_3(a_2/s_3)=a+a_3+a_2=s_1
\end{align}
(\ref{specM}) This formula can be derived without computing the full characteristic polynomial, by observing that since $1$ is an eigenvalue of $p$, the trace and determinant of $p$ are 1 plus the sum of the remaining eigenvalues and the product of these values, respectively. The two cases are of course distinguished by the sign of the discriminant of the quadratic equation.

(\ref{l2}) In this case, the constant term of the quadratic polynomial in (\ref{specM}) is the product $\lambda_2\overline{\lambda_2}=|\lambda_2|^2$. Then $g=1-|\lambda_2|$. \qedsymbol{}

\begin{figure}[h]
  \centering
  \includegraphics[width=0.98\textwidth]{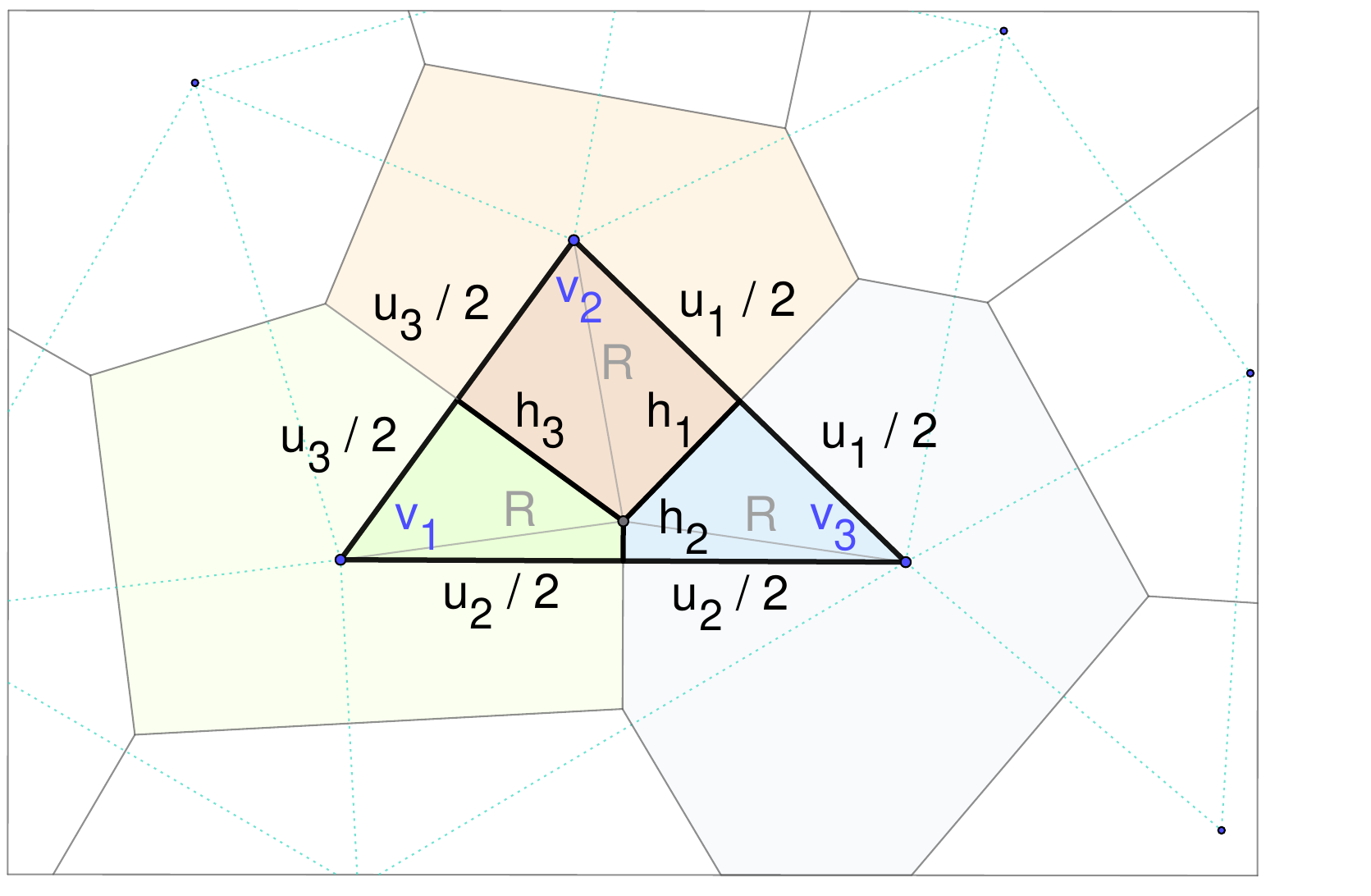}
  \caption{\label{fig:trianglecell}Illustration of the 3-point model.}
\end{figure}

\section{Preliminary investigation of higher-order eigenfunctions of $p^{VM}$} \label{eigensec}

Plots of the first few eigenfunctions of the Voronoi-Markov chain for the concentrated disc are shown in \Cref{fig:eigens}.

\begin{figure}[h]
  \centering
  \includegraphics[width=0.98\textwidth]{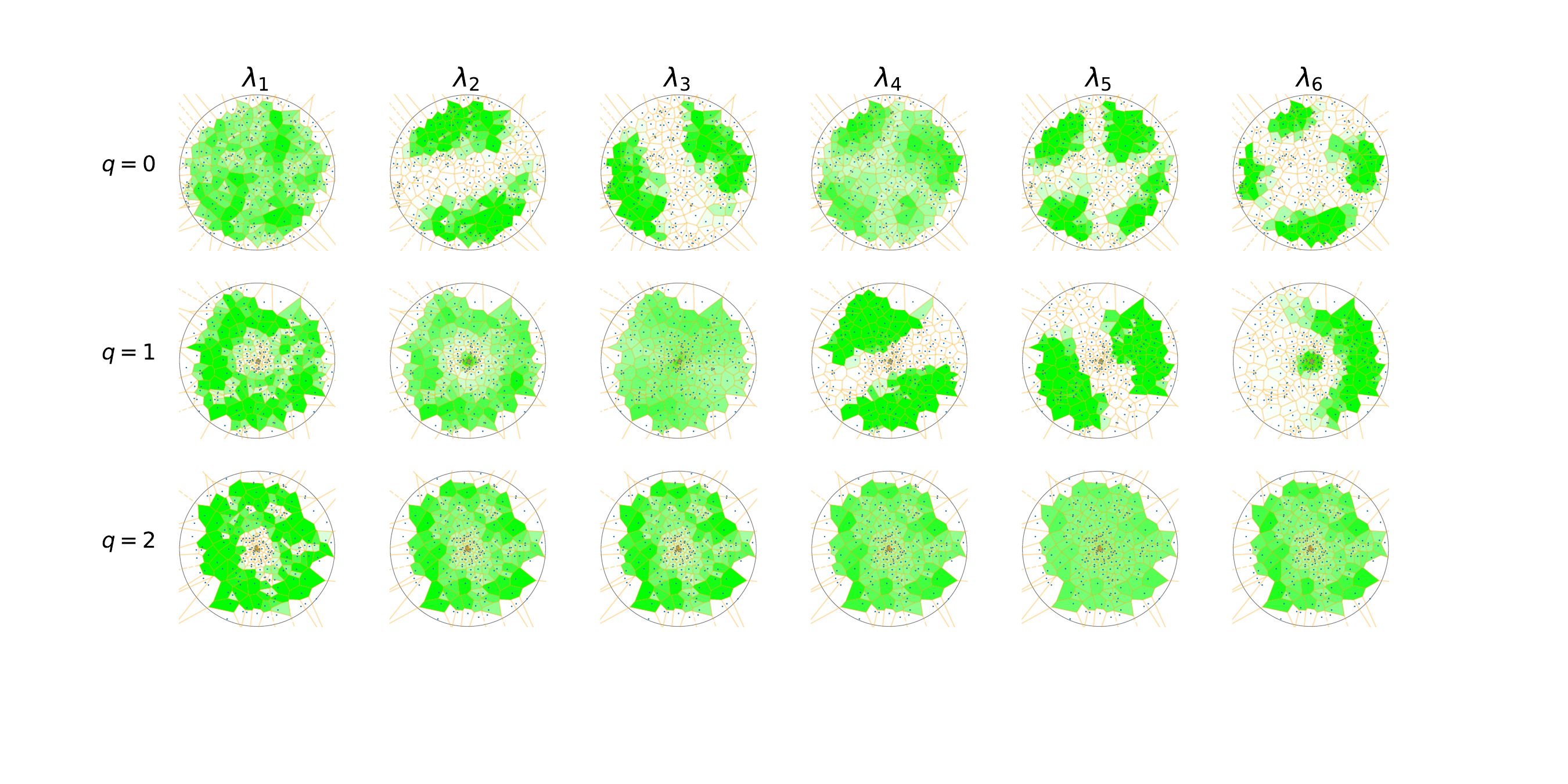}
  \caption{\label{fig:eigens}For the concentrated disc model at 3 values of the concentration $q$, the absolute values of the first 6 eigenfunctions of the Voronoi-Markov chain are plotted, with larger values in darker green. Behavior resembling the harmonics of the Laplacian is apparent in case $q=0$, with the resonances decreasing in variation as $q$ increases (and as the spectral gap decreases). Already at $q=1$ the concentration at the origin is beginning to ``interfere'', most visibly in the case $\lambda_6$, and it appears that the $\lambda_2$ and $\lambda_3$ modes from $q=0$ are shifted to $\lambda_4$ and $\lambda_5$. A sigmoid function has been applied to the eigenfunction values in order to make the differences apparent.}
\end{figure}

\end{document}